\documentclass[aps,pra,twocolumn,groupedaddress,showpacs,showkeys, amsfonts,amssymb]{revtex4}
\usepackage{graphics}
\usepackage{epsfig}
\usepackage{amsmath}
\usepackage{amssymb}
\begin{document}
\newcommand{\hrho}{\widehat{\rho}}
\newcommand{\hsigma}{\widehat{\sigma}}
\newcommand{\homega}{\widehat{\omega}}
\newcommand{\hI}{\widehat{I}}
\newcommand*{\spr}[2]{\langle #1 | #2 \rangle}
\newcommand*{\bbN}{\mathbb{N}}
\newcommand*{\bbR}{\mathbb{R}}
\newcommand*{\cB}{\mathcal{B}}
\newcommand*{\barpi}{\overline{\pi}}
\newcommand*{\barP}{\overline{P}}
\newcommand*{\eps}{\varepsilon}
\newcommand*{\id}{I}
\newcommand{\orho}{\overline{\rho}}
\newcommand{\omu}{\overline{\mu}}
\newcommand*{\half}{{\frac{1}{2}}}
\newcommand*{\ket}[1]{| #1 \rangle}
\newcommand{\trho}{{\widetilde{\rho}}_n^\gamma}
\newcommand{\ttrho}{{\widetilde{\rho}}}
\newcommand{\tsigma}{{\widetilde{\sigma}}}
\newcommand{\tpi}{{\widetilde{\pi}}}
\newcommand*{\bra}[1]{\langle #1 |}
\newcommand*{\proj}[1]{\ket{#1}\bra{#1}}
\newcommand{\otrho}{{\widetilde{\rho}}_n^{\gamma 0}}
\newcommand{\be}{\begin{equation}}
\newcommand{\bea}{\begin{eqnarray}}
\newcommand{\eea}{\end{eqnarray}}
\newcommand{\tr}{\mathrm{Tr}}
\newcommand*{\Hmin}{H_{\min}}
\newcommand{\rank}{\mathrm{rank}}
\newcommand{\tends}{\rightarrow}
\newcommand{\uS}{\underline{S}}
\newcommand{\oS}{\overline{S}}
\newcommand{\uD}{\underline{D}}
\newcommand{\oD}{\overline{D}}
\newcommand{\ee}{\end{equation}}
\newcommand{\supp}{\rm{supp}}
\newcommand{\n}{{(n)}}
\newtheorem{definition}{Definition}
\newtheorem{theorem}{Theorem}
\newtheorem{proposition}{Proposition}
\newtheorem{lemma}{Lemma}
\newtheorem{defn}{Definition}
\newtheorem{corollary}{Corollary}
\newcommand{\qed}{\hspace*{\fill}\rule{2.5mm}{2.5mm}}

\newenvironment{proof}{\noindent{\it Proof}\hspace*{1ex}}{\qed\medskip}
\def\reff#1{(\ref{#1})}
%


\title{Min- and Max- Relative Entropies and a New Entanglement Monotone}


\author{Nilanjana Datta}
\email{N.Datta@statslab.cam.ac.uk}
\affiliation{Statistical Laboratory, DPMMS, University of Cambridge, Cambridge CB3 0WB, UK}

\date{\today}

\begin{abstract}
Two new relative entropy quantities, called the {\em{min- and max-relative 
entropies}},
are introduced and their properties are investigated. The well-known min- and 
max- entropies, introduced by Renner \cite{renatophd}, are obtained 
from these. We 
define a new entanglement monotone, which we refer to as the 
{\em{max-relative entropy of entanglement}}, and which is an upper bound to the relative entropy of entanglement. We also
generalize the min- and max-relative 
entropies to obtain {\em{smooth min-
and max- relative entropies}}. These act as parent quantities for the {\em{smooth
R\'enyi entropies}} \cite{renatophd}, and allow us to define the analogues of the 
mutual information,
in the Smooth R\'enyi Entropy framework. Further, the spectral divergence rates of
the Information Spectrum approach are shown to be obtained from the smooth min- and max-relative
entropies in the asymptotic limit.
\end{abstract}

\bigskip

\noindent

\maketitle


\section{Introduction}
One of the fundamental quantities in Quantum Information Theory is the relative entropy
between two states. 
Other entropic quantities, such as the von Neumann entropy of a state, the conditional entropy
and the mutual information for a bipartite state, are obtainable from the relative entropy. 
Many basic properties of these entropic quantities can be derived from those of the relative entropy.
The strong subadditivity of the
von Neumann entropy, which is one of the most powerful results in Quantum Information
Theory, follows easily from the monotonicity
of the relative entropy. Other than acting as a parent quantity for other entropic quantities, the
relative entropy itself has an operational meaning. It serves as a measure of distinguishability
between different states.

The notion of relative entropy was introduced in 1951, in Mathematical Statistics, by Kullback and Leibler 
\cite{kull}, 
as a means of comparing two different probability distributions. Its extension to the quantum 
setting was
due to Umegaki \cite{umegaki}. The classical relative entropy plays a role similar to its
quantum counterpart. Classical entropic quantities such as the Shannon entropy of a random variable,
the conditional entropy, the mutual information and the joint entropy of a pair of random variables are
all obtainable from it. 

More recently, the concept of relative entropy has been generalized to sequences of states, in the
so-called Information Spectrum Approach \cite{hanverdu93,verdu94,han, ogawa00, nagaoka02, bd1}. The latter is a powerful method which enables us to evaluate
the optimal rates of various information theory protocols, without making any assumption on the
structure of the sources, channels or (in the quantum case) the entanglement resources involved. 
In particular, it allows us to eliminate the frequently-used, but often unjustified,
assumption that sources, channels and entanglement resources are memoryless. 
The quantities arising from the generalizations of the relative entropy in this approach 
are referred to as {\em{spectral divergence rates}}
[see Section \ref{spectral} for their definitions and properties]. Like the relative entropy,
they yield quantities which can be viewed as generalizations of entropy rates 
for sequences of states (or probability distributions,
in the classical case). These quantities have been proved 
to be of important operational significance in Classical and Quantum Information Theory, 
as the optimal rates of protocols such as data compression, 
dense coding, entanglement concentration and dilution, transmission of classical information through
a quantum channel and in the context of hypothesis testing [see e.g.~\cite{ogawa00, nagaoka02, hay_conc, hayashi03, bd_rev, bd_entpure, bd_entmixed, Mat07}]. Hence, spectral divergence rates can be viewed
as the basic tools of a unifying mathematical framework for studying information theoretical protocols.

A simultaneous but independent approach, developed to overcome the limitation of the memoryless criterion 
is the so-called {\em{Smooth Entropy framework}}, developed by Renner et al.~(see e.g. \cite{renatophd},
\cite{KR}, \cite{RenWol04b}, \cite{RenWol05b}, \cite{ReWoWu07}). This approach 
introduced new entropy measures called {\em{smooth R\'enyi entropies}} or {\em{smooth min- and max- entropies}}. 
In contrast to the spectral entropy rates, the (unconditional and conditional) smooth
min- and max- entropies are defined for individual states (or probability distributions) rather than sequences
of states. They are non-asymptotic in nature but depend on an additional parameter $\epsilon$, the smoothness
parameter. Similar to the spectral entropies, the min- and max- entropies have various interesting properties e.g.chain rule inequalities and strong subadditivity. They are also of operational significance and have proved
useful in the context of randomness extraction and cryptography.

Recently it was shown \cite{ndrr} that the two approaches discussed above, 
are related in the sense that the spectral entropy rates 
are obtained as asymptotic limits of the corresponding smooth min- and max- entropies.

In this paper we introduce two new relative entropy quantities, namely the {\em{min- and max- relative entropies}}
(and their smoothed versions), which act as parent quantities for the unconditional and conditional 
min- and max- entropies of Renner \cite{renatophd}. These new relative entropy quantities are seen to
satisfy several interesting properties. Their relations to the quantum relative entropy \cite{OP} and to the quantum
Chernoff bound \cite{chernoff} are discussed. They also allow us to define analogues of the mutual information in the {\em{Smooth R\'enyi Entropy Framework}}. The operational significance of the latter 
will be elaborated in a forthcoming paper. 
We define a new entanglement monotone, which we refer to as the 
{\em{max-relative entropy of entanglement}}, and which is an upper bound to the relative entropy of entanglement.
Moreover, the smooth min- and max- relative
entropies and the analogous quantities in the Quantum Information Spectrum framework, namely the 
spectral divergence rates, are proved to be related in the asymptotic limit.
The proofs are entirely self-contained, relying only on the definitions 
of the entropic quantities involved, and the lemmas stated in Section \ref{math}.

The min- and max- relative entropies both have interesting 
operational significances.
The operational meaning of the min- relative entropy is given in state
discrimination as the negative logarithm of the optimal error probability
of the second kind, when the error probability of the first kind is
required to be zero. This is explained in the proof of Lemma \ref{lem122}.
The max- relative entropy, on the other hand is related to the 
optimal Bayesian error probability, in determining which one of a finite
number (say $M$) of known states a given quantum system is prepared in. 
Suppose the quantum system is prepared in the $k^{th}$ state, $\rho_k$,
with apriori probability $p_k$, and the optimisation
is over all possible choices of POVMs which could be made on the system
to determine its state. Then the optimal Bayesian probability of error
is given by
$$P_{av} = 1 - \inf_{\sigma} \max_{1\le k\le M} p_k 2^{D_{max}(\rho_k||\sigma)},
$$
where the infimum is taken over all possible quantum states, $\sigma$,
in the Hilbert space of the system. This operational interpretation, was first provided in \cite{KRS} though in a somewhat different formalism. It is also explained in \cite{MMND}.

We start with some mathematical preliminaries in Section \ref{math}. We define ({\em{non-smooth}}) min- and max-
relative entropies in Section \ref{non-smooth}, and state how the unconditional and conditional 
min- and max- entropies are obtained from them. We also define the min- and max- mutual 
informations. In Section \ref{props} we investigate the properties of the new relative entropy quantities.
A new entanglement monotone is introduced in Theorem 1 of Section \ref{entm}, and
some of its properties are discussed. Next, we define the smoothed versions of the min- and max- relative
entropies in Section \ref{smooth}. After briefly recalling  the definitions and basic properties
of the spectral divergence rates in Section \ref{spectral}, we go on to 
prove the relations between them and the smooth min- and max- relative entropies in
Section \ref{thms}. These are stated as Theorem 2 and Theorem 3, 
which along with the new entanglement monotone (Theorem 1), 
and the properties of the 
min- and max- relative entropies, constitute the main results of this paper. Our results
apply to the quantum setting and thus include the classical setting as a special case.

\section{Mathematical Preliminaries}
\label{math}

Let ${\cal{B}}({\cal{H}})$ denote the algebra of linear operators
acting on a finite-dimensional Hilbert space ${\cal{H}}$. The von
Neumann entropy of a state $\rho$, i.e., a positive operator of unit
trace in ${\cal{B}}({\cal{H}})$, is given by $S(\rho) = - \tr \rho
\log \rho$. Throughout this paper, we take the logarithm to base $2$
and all Hilbert spaces considered are finite-dimensional. We denote
the identity operator in ${\cal{B}}({\cal{H}})$ by $I$. 

In this paper we make extensive use of spectral projections. Any self-adjoint operator $A$ acting on a finite-dimensional Hilbert space may be written in its spectral
decomposition $A = \sum_i \lambda_i P_i$, where $P_i$ denotes the orthogonal projector
onto the eigenspace of $A$ spanned by eigenvectors corresponding to the eigenvalue $\lambda_i$.  
We define the
positive spectral projection on $A$ as $\{ A \geq 0 \} := \sum_{\lambda_i \geq 0} P_i$, the projector onto the eigenspace of $A$ corresponding to positive eigenvalues.  Corresponding definitions apply for the
other spectral projections $\{ A < 0 \}, \{ A > 0 \}$ and $\{ A \leq 0 \}$. For two operators $A$ and $B$, we can
then define $\{ A \geq B \}$ as $\{ A - B \geq 0 \}$.  The following key lemmas are useful. For a proof of Lemma \ref{lem1} and Lemma \ref{lemmaCPT}, see \cite{bd1, ogawa00,nagaoka02}. Lemma 2
is proved in \cite{ndrr}.
\begin{lemma}
\label{lem1}
For self-adjoint operators $A$, $B$ and any positive operator $0 \leq P 
\leq I,$
we have
\bea
\mathrm{Tr}\big[ P(A-B)\big] &\leq& \mathrm{Tr}\big[ \big\{ A \geq B \big\}
(A-B)\big]\label{lem11}\\
\mathrm{Tr}\big[ P(A-B)\big] &\geq& \mathrm{Tr}\big[ \big\{ A \leq B \big\}
(A-B)\big].
\label{lem12}
\eea
Identical conditions hold for strict inequalities in the spectral
projections $\{A < B\}$ and $\{ A > B\}$.
\end{lemma}
\begin{lemma}
\label{lem2}
Given a state $\rho_n$ and a self-adjoint
operator $\omega_n$, for any real $\gamma$, we have
$$
\mathrm{Tr}\big[\{\rho_n \ge 2^{n\gamma}\omega_n \} \omega_n \bigr]
\leq 2^{-n\gamma}.
$$
\end{lemma}

\begin{lemma}
\label{lemmaCPT}
For self-adjoint operators $A$ and $B$, and any completely positive trace-preserving (CPTP) map $\mathcal{T}$, the inequality
\begin{equation}
\mathrm{Tr}\big[ \{\mathcal{T}(A) \geq \mathcal{T}(B) \}\mathcal{T}(A-B)\big] \leq \mathrm{Tr}\big[ \big\{ A \geq B \big\} (A-B)\big]
\label{eqn:second_ineq}
\end{equation}
holds.
\end{lemma}

The trace distance between two operators $A$ and $B$ is given by
\be
||A-B||_1 := \tr\bigl[\{A \ge B\}(A-B)\bigr] -
 \tr\bigl[\{A < B\}(A-B)\bigr]
\ee
The fidelity of states $\rho$ and $\rho'$ is defined to be
$$ F(\rho, \rho'):= \tr \sqrt{\rho^{\half} \rho' \rho^{\half}}.
$$
The trace distance between two states $\rho$ and $\rho'$ is
related to the fidelity $ F(\rho, \rho')$ as follows (see (9.110) of \cite{nielsen}):
\be
  \frac{1}{2} \| \rho - \rho' \|_1
\leq
  \sqrt{1-F(\rho, \rho')^2}
\leq
  \sqrt{2(1-F(\rho, \rho'))} \ .
\label{fidelity}
\ee

We use the following simple corollary of Lemma \ref{lem1}:
\begin{corollary}
\label{cor1}
For self-adjoint operators $A$, $B$ and any positive operator $0 \leq P \leq
I$,
the inequality
$$||A-B||_1 \le \eps,$$
for any $\eps >0$,
implies that $$\tr \bigl[P(A-B) \bigr] \le \eps.$$
\end{corollary}
We also use the ``gentle measurement'' lemma \cite{winter99,ogawanagaoka02}:
\begin{lemma}\label{gm} For a state $\rho$ and operator $0\le \Lambda\le I$, if
$\mathrm{Tr}(\rho \Lambda) \ge 1 - \delta$, then
$$||\rho -   {\sqrt{\Lambda}}\rho{\sqrt{\Lambda}}||_1 \le {2\sqrt{\delta}}.$$
The same holds if $\rho$ is only a subnormalized density operator.
\end{lemma}

\section{Definitions of min- and max- relative entropies}
\label{non-smooth}

\begin{definition}
  The \emph{max- relative entropy} of two operators
  $\rho$ and $\sigma$, such that $\rho \ge 0$, $\tr \rho \le 1$ and 
$\sigma \ge 0$, is defined by  
\be
    D_{\max}(\rho|| \sigma)
  :=
    \log \min\{ \lambda: \, \rho\leq \lambda \sigma \}
  \label{dmax}
\ee
\end{definition}
Note that $D_{\max}(\rho|| \sigma)$ is well-defined if 
$\supp\, \rho \subseteq \supp\, \sigma$. For $\rho$ and $\sigma$
satisfying $\supp\, \rho \subseteq \supp\, \sigma$, $ D_{\max}(\rho|| \sigma)$
is equivalently given by
 \be
    D_{\max}(\rho|| \sigma)
  :=
    \log \mu_{\max}\bigl(\sigma^{-\frac{1}{2}}\rho\sigma^{-\frac{1}{2}}\bigr),
  \label{dmax2}
\ee
where the notation $\mu_{\max}(A)$ is used to denote the maximum eigenvalue
of the operator $A$, and the inverses are generalized inverses defined as follows:
$A^{-1}$ is a generalized inverse of $A$ if $AA^{-1} = A^{-1}A =P_A = P_{A^{-1}}$,
where $P_A, P_{A^{-1}}$ denote the projectors onto the supports of $A$ and $A^{-1}$ respectively.

Another equivalent definition of $D_{\max}(\rho|| \sigma)$ is:
\be
D_{\max}(\rho|| \sigma):=\log \min\{ \lambda: \, \tr\bigl[P_+^\lambda (\rho - \lambda \sigma) \bigr]=0\},
\label{dmax3}
\ee
where $P_+^\lambda:= \{\rho \ge \lambda \sigma\}$.

\begin{definition}
  The \emph{min- relative entropy} of two operators
  $\rho$ and $\sigma$, such that $\rho \ge 0$, $\tr \rho \le 1$ and 
$\sigma \ge 0$, is defined by  
  \be
    D_{\min}(\rho|| \sigma)
  :=    - \log \tr\bigl(\pi_\rho\sigma\bigr) \ ,
  \label{dmin}
\ee
  where $\pi_\rho$ denotes the projector onto $\supp\, \rho$, the support of $\rho$. It is well-defined if 
$\supp\, \rho$ has non-zero intersection with $\supp\, \sigma$.
\end{definition}

Note that  
\be
D_{\min}(\rho|| \sigma) = \lim_{\alpha\rightarrow 0^+} S_\alpha(\rho||\sigma),
\label{relren}
\ee
where $S_\alpha(\rho||\sigma)$ denotes the {\em{quantum relative R\'enyi entropy}} of order $\alpha$, with
$0<\alpha<1$, defined by (see e.g. \cite{OP, hayashibook}):
\be
S_\alpha(\rho||\sigma):= \frac{1}{\alpha-1}\log \tr \rho^{\alpha}\sigma^{1-\alpha}.
\label{renyi}
\ee
Various properties of $D_{\min}(\rho|| \sigma)$ and $D_{\max}(\rho|| \sigma)$ are discussed in
Section \ref{props}.

The min- and max- (unconditional and conditional) entropies, introduced 
by Renner
in \cite{renatophd} can be obtained from $D_{\min}(\rho|| \sigma)$ and $D_{\max}(\rho|| \sigma)$
by making suitable substitutions for the positive operator $\sigma$. In particular,
for $\sigma = I$, we obtain the min- and max- entropies of a state $\rho$, which are
simply the R\'enyi entropies of order infinity and zero, respectively:
\be 
H_{\min} (\rho) = -D_{\max}(\rho|| I) = - \log \| \rho\|_{\infty}
\label{minent}
\ee
\be
H_{\max} (\rho) = -D_{\min}(\rho|| I) = \log \rank(\rho).
\ee
The {min-} and {max-entropies} of a bipartite state, $\rho_{A B}$, relative to a state $\sigma_B$, are 
similarly obtained by setting $\sigma= I_A\otimes \sigma_B$:
\bea
    H_{\min}(\rho_{A B} | \sigma_B)
  &:=& 
    - \log \min\{ \lambda: \, \rho_{A B} \leq \lambda \cdot \id_A \otimes \rho_B \}\nonumber\\
&=& -D_{\max}(\rho_{AB} || I_A\otimes \sigma_B) 
\eea
  and
\bea
    H_{\max}( \rho_{A B} | \sigma_B )
  &:=& 
    \log \tr\bigl(\pi_{A B} (\id_A \otimes \sigma_B)\bigr) \ ,\nonumber\\
  &=& -D_{\min}(\rho_{AB} || I_A\otimes \sigma_B) 
\eea
In the above, $\pi_{A B}$ denotes the projector onto the support of $\rho_{A B}$.

In addition, by considering $\sigma = \rho_A \otimes \rho_B$, 
we obtain the following analogues of the quantum mutual information of a bipartite
state $\rho = \rho_{AB}$:

\begin{definition} For a bipartite state $\rho_{AB}$, the min- and max- mutual
informations are defined by
\bea
D_{\min}(A:B)&:=& D_{\min} (\rho_{AB} || \rho_A\otimes \rho_B)\nonumber\\
D_{\max}(A:B)&:=& D_{\max} (\rho_{AB} || \rho_A\otimes \rho_B)\nonumber\\
\label{mutual}
\eea  
\end{definition}

Smooth versions of the min- and max- relative entropies are defined in Section \ref{smooth}.
These in turn yield the (unconditional and conditional) smooth min- and max- entropies \cite{renatophd,
ndrr} and mutual informations, upon similar substitutions for the operator $\sigma$. It was proved 
in \cite{ndrr} that the smooth min- and max- entropies are related to the spectral entropy rates
used in the Quantum Information Spectrum approach [see Section \ref{spectral} or \cite{bd1}],
in the sense that the spectral entropy rates are the asymptotic limits of the
smooth entropies. As discussed in Section \ref{spectral}, the spectral entropy rates
are obtainable from two quantities, namely the inf- and sup- spectral divergence rates.
In Section \ref{thms}, we prove that these spectral divergence rates are indeed asymptotic limits of the
smooth min- and max- relative entropies.

\section{Properties of min- and max- relative entropies}
\label{props}
The min- and max- relative entropies satisfy the following properties:
\begin{lemma} 
\label{minmax} For a state $\rho$ and a positive operator $\sigma$
\be D_{\min}(\rho||\sigma) \le D_{\max}(\rho||\sigma)
\ee
\end{lemma}
\begin{proof} (This is exactly analogous to the proof of Lemma 3.1.5 in \cite{renatophd}).
Let $\pi_\rho$ denote the projector onto the support of $\rho$, and let $\lambda \ge 0$
such that $D_{\max}(\rho||\sigma) = \log \lambda$, i.e., $\lambda \sigma - \rho \ge 0$.
Then, using the fact that for positive semi-definite operators $A$ and $B$, $\tr(AB) \ge 0$,  
we get
$$
0 \le \tr \bigl((\lambda \sigma - \rho)\pi_\rho\bigr) = \lambda \tr (\pi_\rho \sigma) - 1.$$
Hence, 
$$
D_{\min}(\rho||\sigma) := - \log \tr (\pi_\rho \sigma) \le \log \lambda = D_{\max}(\rho||\sigma) 
$$
\end{proof}

\begin{lemma} 
\label{lem6}
The min- and max- relative entropies are non-negative when both $\rho$ and $\sigma$ are 
states. They are both equal to zero when $\rho$ and $\sigma$ are identical states. Moreover, 
$D_{\min}(\rho||\sigma)=0$ when $\rho$ and $\sigma$ have identical supports.
\end{lemma}
\begin{proof} Due to Lemma \ref{minmax}, it suffices to prove that $D_{\min}(\rho||\sigma) \ge0$,
when $\rho$ and $\sigma$ are states.
Note that $\tr (\pi_\rho \sigma) \le \tr \sigma = 1$, where $\pi_\rho$ denotes the projector onto the 
support of $\rho$. Hence, 
$$D_{\min}(\rho||\sigma):= - \log \tr (\pi_\rho \sigma) \ge 0.$$

The rest of the lemma follows directly from the definitions \reff{dmax2} and \reff{dmin} 
of the max- and min- relative entropies, respectively.

\end{proof}

\begin{lemma}
\label{mono} 
The min- and max- relative entropies are monotonic under CPTP maps, i.e., for
a state $\rho$, a positive operator $\sigma$, and a CPTP map $\mathcal{T}$:
\be D_{\min}(\mathcal{T}(\rho)||\mathcal{T}(\sigma)) \le    D_{\min}(\rho||\sigma)
\label{monmin}
\ee
and
\be D_{\max}(\mathcal{T}(\rho)||\mathcal{T}(\sigma))\le  D_{\max}(\rho||\sigma)
\label{monmax}
\ee
\end{lemma}
\begin{proof}
The monotonicity \reff{monmin} follows directly from the monotonicity of the quantum 
relative R\'enyi entropy. For $0<\alpha<1$, we have \cite{hayashibook}:
$$S_\alpha(\rho||\sigma) \le S_\alpha(\mathcal{T}(\rho)||\mathcal{T}(\sigma)).$$
Taking the limit $\alpha \rightarrow 0^+$ on both sides of this inequality 
and using \reff{relren}, yields \reff{monmin}.

The proof of \reff{monmax} is analogous to Lemma 3.1.12 of \cite{renatophd}.
Let $\lambda \ge 0$ such that $\log \lambda = D_{\max}(\rho||\sigma)$ and 
hence $(\lambda\sigma - \rho) \ge 0$. Since $\mathcal{T}$ is a CPTP map,
$\mathcal{T}(\lambda\sigma - \rho)= \lambda\mathcal{T}(\sigma) - \mathcal{T}(\rho) \ge 0.$
Hence, 
\bea
D_{\max}(\mathcal{T}(\rho)||\mathcal{T}(\sigma)) &:=& \log \min \{ \lambda^\prime : 
\mathcal{T}(\rho) 
\le \lambda^\prime \mathcal{T}(\sigma)\}\nonumber\\
&\le& \log \lambda = D_{\max}(\rho||\sigma).\eea
\end{proof}

\begin{lemma} The min-relative entropy is jointly convex
in its arguments.
\end{lemma}
\begin{proof} The proof follows from the monotonicity of the min-relative entropy under CPTP maps (Lemma \ref{mono}). 
Following \cite{hayashibook}, let $\rho_1, \ldots \rho_n$
be states acting on a Hilbert space ${\cal{H}}$, let $\sigma_1, \ldots \sigma_n$
be positive operators in ${\cal{B}}({\cal{H}})$ such that 
$\supp \,\rho_i \subseteq \supp\,\sigma_i$, for $i=1,\ldots,n$, and let ${\displaystyle{\{p_i\}_{i=1}^n}}$ denote a 
probability distribution. Let $\rho, \sigma \in {\cal{B}}({\cal{H}})$.

Consider the following operators in ${\cal{B}}({\cal{H}} \otimes {\bf{C}}^n)$
$$
A:= \sum_{i=1}^n p_i A_i =\sum_{i=1}^n p_i |i\rangle \langle i| \otimes \rho_i,
$$
$$
B:= \sum_{i=1}^n p_i B_i= \sum_{i=1}^n p_i |i\rangle \langle i|\otimes \sigma_i$$
Note that the the operators $A_i$, $B_j$, $i,j \in \{1,2,\ldots n\}$,
have orthogonal support for $i\ne j$, i.e.,
\be \tr A_i A_j = 0 = \tr A_i B_j= \tr B_i B_j \quad {\hbox{for }} i \ne j.
\label{ortho} \ee
Let $\pi_A$ and $\pi_i$ denote the orthogonal projectors onto the support of 
$A$ and $\rho_i$, respectively, with $i=1,2,\ldots n$. Then using \reff{ortho} 
and the convexity of the function $- \log x$, we obtain
\bea
D_{\min} (A||B) &=& - \log \tr \bigl( \pi_A B\bigr)= - \log  \bigl(\sum_{i=1}^n p_i 
\tr (\pi_i \sigma_i)\bigr)\nonumber\\
&\le &  \sum_{i=1}^n p_i\bigl[- \log \tr (\pi_i \sigma_i)\bigr] \nonumber\\
&=& 
\sum_{i=1}^n p_i D_{\min}(\rho_i || \sigma_i).
\label{rf1}
\eea
Taking the partial traces of $A$ and $B$ over ${\bf{C}}^n$ yields the operators
$$ \tr_{{\bf{C}}^n} A = \sum_i p_i \rho_i \quad ; \quad
\tr_{{\bf{C}}^n} B = \sum_i p_i \sigma_i.$$
However, since the partial trace over ${\bf{C}}^n$ is a CPTP map, we have by Lemma \ref{mono} that
\be
 D_{\min}\Bigl(\sum_{i=1}^n p_i \rho_i || \sum_{i=1}^n p_i \sigma_i\Bigr) \le  D_{\min} (A||B)
\label{rf2}
\ee
The inequalities \reff{rf1} and \reff{rf2} yield the 
joint convexities:
\be
D_{\min}\Bigl(\sum_{i=1}^n p_i \rho_i || \sum_{i=1}^n p_i \sigma_i\Bigr) \le \sum_{i=1}^n p_i D_{\min}(\rho_i || \sigma_i)
\label{rf3}
\ee
\end{proof}

\begin{lemma}
The max- relative entropy of two mixtures of states, $\rho:=\sum_{i=1}^n p_i \rho_i$ and
$\sigma:= \sum_{i=1}^n p_i \sigma_i$, satisfies the following bound:
\be
D_{\max}(\rho||\sigma) \le \max_{1\le i\le n} D_{\max}(\rho_i || \sigma_i)
\label{rf4}
\ee
\end{lemma}

\begin{proof}
By definition \reff{dmax3}: $$
D_{\max}(\rho||\sigma) = \log \min \{\lambda : \tr \bigl[P_+^\lambda (\rho - \lambda \sigma)\bigr]=0\},$$ 
where $P_+^\lambda = \{\rho \ge \lambda \sigma\}$. Consider the projection operators
$P_+^{\lambda,i} := \{\rho_i \ge \lambda \sigma_i\}$ for $i=1,2,\ldots, n$. 
Then
\bea
0\le \tr \bigl[P_+^\lambda (\rho - \lambda \sigma)\bigr]&=& \sum_i p_i \tr \bigl[P_+^{\lambda} (\rho_i - \lambda \sigma_i)\bigr]
\nonumber\\
&\le &\sum_i p_i \tr \bigl[P_+^{\lambda,i} (\rho_i - \lambda \sigma_i)\bigr], 
\label{rhs}
\nonumber\\
\eea
by Lemma \ref{lem1}.
Set $\lambda = \max_{1\le i\le n} \lambda_i$ where for each $i=1,2,\ldots, n$, $ \lambda_i$ is defined by
$$ \log \lambda_i = D_{\max}(\rho_i || \sigma_i).$$ For this choice of $\lambda$, each term 
in the sum on the right hand side of
\reff{rhs} vanishes, implying that $\tr \bigl[P_+^\lambda (\rho - \lambda \sigma)\bigr]=0$, and hence
$\lambda \ge D_{\max}(\rho||\sigma)$. 
\end{proof}

\begin{lemma}
\label{best}
The min- and max- relative entropies of two states $\rho$ and $\sigma$ 
are related to the quantum relative entropy
$S(\rho||\sigma):= \tr \bigl[\rho \log \rho - \rho \log \sigma\bigr]$ as 
follows:
\be
D_{\min}(\rho||\sigma) \le S(\rho||\sigma) \le D_{\max}(\rho||\sigma). 
\label{max1}
\ee
\end{lemma}

\begin{proof}
We first prove the upper bound $S(\rho||\sigma) \le D_{\max}(\rho||\sigma)$:

Let $\rho \le 2^\alpha \sigma$, with $\alpha = D_{\max}(\rho||\sigma)$. Then 
using the operator monotonicity of the logarithm \cite{bhatia}, we have 
$\log \rho \le \alpha + \log \sigma$. This in turn implies that 
$\rho \log \sigma \ge \rho \log \rho - \alpha \rho.$
Hence, for a state $\rho$, $\tr \rho \log \sigma \ge \tr \rho \log 
\rho - \alpha,$ and 
\bea
S(\rho||\sigma) &:=& \tr \rho \log \rho - \tr \rho \log \sigma \nonumber\\
&\le &  \tr \rho \log \rho - \tr \rho \log \rho + \alpha\nonumber\\
&=&   D_{\max}(\rho||\sigma).
\eea
\medskip

\noindent
We next prove the bound $D_{\min}(\rho||\sigma) \le S(\rho||\sigma)$:

Consider the CPTP map, ${\cal{T}}$, defined by
$${\cal{T}}(\omega) = \pi_\rho \omega \pi_\rho + \barpi \omega \barpi,$$
where $\omega$ is any density matrix, $\pi_\rho$ is the projector
onto the support of $\rho$, and ${\overline{\pi_\rho}} = I - \pi_\rho.$
Note that ${\cal{T}}(\rho)=\rho$.

Due to the monotonicity of $S(\rho||\sigma)$ under CPTP maps, we have
\bea
S(\rho||\sigma) &\ge& S({\cal{T}}(\rho)||{\cal{T}}(\sigma))\nonumber\\
&=& \tr \rho \log \rho - \tr \rho \log (\pi_\rho \sigma \pi_\rho)\nonumber\\
&=& S(\rho|| \pi_\rho \sigma \pi_\rho).
\label{r1}
\eea
Define the normalized state ${\displaystyle{\tsigma:= \frac{1}{c}\pi_\rho \sigma \pi_\rho}}$ where
$c = {\tr(\pi_\rho \sigma)}$. Then
\bea
S(\rho|| \pi_\rho \sigma \pi_\rho) &=& S(\rho || c \tsigma) \nonumber\\
&=& \tr \rho (\log \rho - \log \tsigma) - (\log c). \tr \rho\nonumber\\
&=& S(\rho||\tsigma) - \log c\nonumber\\
&\ge& - \log c = D_{\min} (\rho || \sigma).
\label{r2}
\eea   
From \reff{r1} and \reff{r2} we conclude that 
$D_{\min} (\rho || \sigma) \le S(\rho|| \sigma )$.
\end{proof}  

The following lemma is obtained easily from the definitions of the min-
and max-relative entropies. 
\begin{lemma}
\label{unitary} The min- and max- relative entropies are invariant 
under joint unitary transformations.
\end{lemma}

\begin{lemma}
\label{lem122}
The min- relative entropy of two states $\rho$ and $\sigma$ 
for which $\supp\, \rho \subseteq \supp\, \sigma$, satisfies the following bounds:
\be
D_{\min}(\rho||\sigma) \le D_{\max}(\rho||\sigma) \le - \log \mu_{\min}(\sigma),
\label{first}
\ee
where $\mu_{\min}(\sigma)$ denotes the minimum non-zero eigenvalue of $\sigma$.
Further,
{\be
D_{\min}(\rho||\sigma) \le - \log \bigl[ 1 - \frac{1}{2} ||\rho-\sigma||_1 \bigr]
\label{two}
\ee}
\end{lemma}
\begin{proof}
The first inequality in \reff{first}
has been proved in Lemma \ref{minmax}.
Since $\supp\, \rho \subseteq \supp\, \sigma$, we have $\pi_\rho \le \pi_\sigma$,
where $\pi_\rho$ and $\pi_\sigma$, denote the projectors onto the supports of 
$\rho$ and $\sigma$ respectively. Further, using the bounds $\rho \le \pi_\rho$ and
$\pi_\sigma \le {\mu_{\min}(\sigma)}^{-1} \sigma$, where $\mu_{\min}(\sigma)$ denotes 
the minimum non-zero eigenvalue of $\sigma$,
we obtain
$$ \rho \le {\mu_{\min}(\sigma)}^{-1} \sigma.$$
Using the definition of $D_{\max}(\rho||\sigma)$ we therefore infer that 
$D_{\max}(\rho||\sigma) \le - \log \mu_{\min}(\sigma)$.

The bound \reff{two} follows from the fact that
\be
\tr (\pi_\rho\sigma) \ge 1 - \frac{1}{2} ||\rho-\sigma||_1,
\label{bdd3}
\ee
which can be seen as follows. Consider the scenario of state discrimination. Suppose it
is known that that a finite quantum system is in one of two states 
$\rho$ and $\sigma$, with equal apriori probability. To determine which state
it is in, one does a binary Postive Operator-Valued Measurement (POVM) with elements $E_1$ and $E_2$,
and $E_1 + E_2= I$. If the outcome corresponding to $E_1$ occurs then the system is inferred to be in the state
$\rho$, whereas if the outcome corresponding to $E_2$ occurs, then the system is inferred to be in the state
$\sigma$. The average probability of error in state discrimination is given by 
\bea
p_e^{av} &=& 1 - \frac{1}{2}( \tr E_1 \rho +\tr E_2 \sigma) \nonumber\\
&=& \frac{1}{2}( \tr E_2 \rho +\tr E_1 \sigma)
\label{hel}
\eea
Note that the two terms in the parenthesis, in the last line of \reff{hel} are, 
respectively, the {\em{Type I error}} and the {\em{Type II error}},
in the language of hypothesis testing.
By Helstrom's Theorem \cite{Helstrom} the minimum possible value of $p_e^{av}$ is given by
$$ \frac{1}{2}\bigl[1 - \frac{1}{2} ||\rho - \sigma||_1\bigr].$$
Now consider a POVM in which $E_1 = \pi_\rho$ (the projector onto the support of $\rho$) and $E_2 = I - \pi_\rho$.
In this case, the {\em{Type I error}} vanishes and $p_e^{av}= \frac{1}{2} \tr (\pi_\rho \sigma)$, which by Helstrom's
theorem satisfies the bound:
$$\frac{1}{2} \tr (\pi_\rho \sigma) \ge \frac{1}{2}\bigl[1 - \frac{1}{2} ||\rho - \sigma||_1\bigr],$$
hence yielding the desired bound \reff{bdd3}. Note that $D_{\min}(\rho||\sigma)$ is therefore 
related to the average probability of error of state discrimination 
(between the states $\rho$ and $\sigma$) when the Type I error vanishes.

Note that \reff{bdd3} can also be proved algebraically by a simple use of Lemma \ref{lem1}.
 \end{proof}

It is known that if one has asymptotically many copies
of two states $\rho$ and $\sigma$, the error in discriminating between them decreases 
exponentially, and the error exponent is given by the so-called
quantum Chernoff bound $\xi(\rho,\sigma)$ \cite{chernoff}. The latter has been shown to be given by
$\xi(\rho,\sigma)  = - \log \bigl( \min_{0\le s\le 1} \tr \rho^s \sigma^{1-s}\bigr)$.
From the expression \reff{relren} it follows that the min-relative entropy provides a lower bound 
to the quantum Chernoff bound, i.e.,
$$\xi(\rho, \sigma) \ge D_{\min}(\rho||\sigma).$$

\section{A new entanglement monotone}
\label{entm}
Here we introduce a new entanglement monotone for a bipartite 
state $\rho$. We denote it by $E_{\max}(\rho)$ and call it the
{\em{max-relative entropy of entanglement}}. 

In \cite{VP}, Vedral and Plenio proved a set of sufficient conditions 
under which
a measure $D(\rho||\sigma)$ of the ``distance'' between
two bipartite quantum states, $\rho$ and $\sigma$,
defines an entanglement monotone, $E(\rho)$, through the following
expression:
$$  E(\rho):= \min_{\sigma \in {\cal{D}}} D(\rho||\sigma).$$
Here the minimum is taken over the set ${\cal{D}}$ of all separable states
(see also \cite{horo}).
The conditions ensure that $(i)$ $E(\rho)=0$ if and only if $\rho$ is
separable;  $(ii)$ $E(\rho)$ is unchanged by a local change of basis and 
 $(iii)$ $E(\rho)$ does not increase on average under local operations 
and classical communication (LOCC) \footnote{If $\Lambda$ denotes a LOCC map
and $\{p_i, \rho_i\}$ denotes an ensemble of states
such that $\rho_i$ is obtained with probability $p_i$ under the action
of $\Lambda$ on $\rho$, then $E(\rho) \ge 
\sum_i p_i E(\rho_i).$}.
By proving that these conditions are satisfied by 
$D_{\max} (\rho||\sigma)$, 
we are able to define a new
entanglement monotone.
\begin{theorem}
\label{ent}
For a bipartite state $\rho$, the quantity 
\be E_{\max}(\rho):= \min_{\sigma \in {\cal{D}}} D_{\max} (\rho||\sigma),
\label{entmeasure}
\ee
where the minimum is taken over the set  ${\cal{D}}$ 
of all separable states,
is an entanglement monotone. \end{theorem}
\begin{proof} 
It suffices to prove that the max-relative entropy 
$D_{\max} (\rho||\sigma)$ satisfies the following properties,
which constitute the set of sufficient conditions 
proved in \cite{VP}. 
\begin{itemize}
\item{$D_{\max} (\rho||\sigma) \ge 0$ with equality if and only if 
$\rho=\sigma$.}
\item{$D_{\max} (\rho||\sigma) = D_{\max} (U\rho U^\dagger||U\sigma U^\dagger)$
for any unitary operator $U$.}
\item{$D_{\max} (\tr_p\rho||\tr_p\sigma)\le D_{\max} (\rho||\sigma)$, where 
$\tr_p$ denotes a partial trace.}
\item{$\sum_i \alpha_i D_{\max} (\rho_i/\alpha_i||\sigma_i/\beta_i)
\le \sum_i D_{\max} (\rho_i||\sigma_i),$ where $\alpha_i = \tr \rho_i$, 
$\beta_i = \tr \sigma_i$, $\rho_i = V_i \rho V_i^\dagger$,
$\sigma_i = V_i \sigma V_i^\dagger$ and $\sum_i V_i^\dagger V_i =1$.}
\item{For any set $\{P_i\}$ of mutually orthogonal projectors, i.e.,
$P_iP_j= \delta_{ij}P_i$, 
\bea
&&
D_{\max} \bigl(\sum_iP_i\rho P_i||\sum_i P_i\sigma P_i\bigr)\nonumber\\
&=& \sum_i D_{\max} \bigl(P_i\rho P_i||P_i\sigma P_i\bigr).
\eea}
\item{For any projector $P$,
$$D_{\max} \bigl(\rho \otimes P||\sigma \otimes P\bigr)
= D_{\max} (\rho||\sigma).$$}
\end{itemize}
From Lemma \ref{lem6} we have that if $\rho=\sigma$ then
$D_{\max} (\rho||\sigma)=0$. To prove the converse, i.e., 
$D_{\max} (\rho||\sigma)=0$ implies that  $\rho=\sigma$, note that
$D_{\max} (\rho||\sigma)=0$ $\implies$ $(\sigma - \rho) \ge 0$. On  
the other hand, since $\rho$ and $\sigma$ are states, we have 
$\tr(\sigma-\rho)=0$, which in turn implies that $(\sigma-\rho)= 0$.
This completes the proof of the first property.
 The second property follows from Lemma 
\ref{unitary}. The third property follows from Lemma \ref{mono},
since the partial trace is a CPTP map. 

The fourth property can be proved as follows. Using the definition
\reff{dmax2} of the max-relative entropy, we have 
\bea
&&
\sum_i \alpha_i D_{\max} (\rho_i/\alpha_i||\sigma_i/\beta_i)\nonumber\\
&=&
\sum_i \alpha_i \log \Bigl[\mu_{\max}\Bigl(
\Bigl(\frac{\sigma_i}{\beta_i}\Bigr)^{-\frac{1}{2}} \frac{\rho_i}{\alpha_i} \Bigl(\frac{\sigma_i}{\beta_i}\Bigr)^{-\frac{1}{2}} \Bigr)\Bigr]\nonumber\\
&=& \sum_i \alpha_i \log \Bigl(\frac{\beta_i}{\alpha_i}\mu_{\max} \Bigl(
\sigma_i^{-\frac{1}{2}} \rho_i \sigma_i^{-\frac{1}{2}}\Bigr)\Bigr)\nonumber\\
&=& \sum_i \alpha_i \log \frac{\beta_i}{\alpha_i} + 
\sum_i \alpha_i \log \mu_{\max} \Bigl(
\sigma_i^{-\frac{1}{2}}\rho_i \sigma_i^{-\frac{1}{2}}\Bigr)\nonumber\\
&\le& -\sum_i \alpha_i \log \frac{\alpha_i}{\beta_i}+ \sum_i \log \mu_{\max} \Bigl(
\sigma_i^{-\frac{1}{2}}\rho_i \sigma_i^{-\frac{1}{2}}\Bigr)\nonumber\\
&\le & \sum_i \log \mu_{\max} \Bigl(
\sigma_i^{-\frac{1}{2}}\rho_i \sigma_i^{-\frac{1}{2}}\Bigr)\nonumber\\
&=& \sum_i D_{\max}(\rho_i||\sigma_i).
\eea
In the above we have made use of the fact that $\{\alpha_i\}$ and
$\{\beta_i\}$ are probability distributions, and hence 
${\displaystyle{\sum_i \alpha_i \log \frac{\alpha_i}{\beta_i}\ge 0}}$ \cite{cover}.
The fifth and sixth properties are seen to hold by inspection.
\end{proof}

\begin{lemma}
\label{upent}
The {\em{max-relative entropy of 
entanglement}}, $E_{\max}$,
is an upper bound to the {\em{relative entropy of entanglement}} \cite{VP}
$$E(\rho) :=\min_{\sigma \in {\cal{D}}} S (\rho||\sigma).$$
\end{lemma}
\begin{proof}
Let $\sigma^*$ be a separable state such that
\bea
D_{\max}(\rho|| \sigma^*)&=&
\min_{\sigma \in {\cal{D}}} D(\rho|| \sigma)\nonumber\\
&=&E_{\max}(\rho)
\eea
By Lemma \ref{best} we have
\bea
D_{\max}(\rho|| \sigma^*)&\ge& S(\rho|| \sigma^*)\nonumber\\
&\ge& \min_{\sigma \in {\cal{D}}} S(\rho|| \sigma)\nonumber\\
&=& E(\rho).
\eea
\end{proof}

Properties of $ E_{\max}(\rho)$ will be investigated
in a forthcoming paper.

\section{Smooth min- and max- relative entropies}
\label{smooth}
\emph{Smooth} min- and max- relative entropies are generalizations of the above-mentioned  
relative entropy measures, involving an additional \emph{smoothness} parameter
$\eps \geq 0$. For $\eps = 0$, they reduce to the
\emph{non-smooth} quantities 
\begin{definition} \label{def:smoothentropies} For
  any $\eps \geq 0$, the \emph{$\eps$-smooth min-} and
  \emph{max-relative entropies} of a bipartite state $\rho$ relative to a
  state $\sigma$ are defined by
\[
    D_{\min}^{\eps}(\rho || \sigma)
  :=
    \sup_{\bar{\rho} \in B^{\eps}(\rho)} D_{\min}(\bar{\rho} || \sigma)
  \]
  and
  \[
    D_{\max}^{\eps}( \rho|| \sigma )
  :=
    \inf_{\bar{\rho} \in B^{\eps}(\rho)} D_{\max}(\bar{\rho} || \sigma)
  \]
  where $B^{\eps}(\rho) := \{\bar{\rho} \geq 0: \, \| \bar{\rho} - \rho
  \|_1 \leq \eps, \tr(\bar{\rho}) \leq \tr(\rho)\}$.
\end{definition}

The following two lemmas are used in the proof of Theorem \ref{thm_divmax} given below.  

\begin{lemma} \label{lem5n}
  Let $\rho_{A B}$ and  $\sigma_{AB}$ be density operators, let $\Delta_{A B}$ be a positive operator, and let $\lambda \in \bbR$  such that
   \[
    \rho_{A B} \leq 2^{\lambda} \cdot \sigma_{AB} + \Delta_{A B} \ .
  \]
  Then $D_{\max}^{\eps}(\rho_{A B}||\sigma_{AB}) \le \lambda$ for any $\eps \geq \sqrt{8 \tr(\Delta_{A B})}$.
\end{lemma}

\begin{lemma} \label{lem6n}
  Let $\rho_{A B}$ and $\sigma_{AB}$ be density operators. Then
  \[
    D_{\max}^{\eps}(\rho_{A B}||\sigma_{AB}) \le \lambda
  \]
  for any $\lambda \in \bbR$ and
  \[
    \eps = \sqrt{8 \tr\bigl[\{\rho_{A B} > 2^{\lambda} \sigma_{AB} \} \rho_{A B} \bigr]} \ .
  \]
\end{lemma}

The proofs of these lemmas are analogous to the proofs of Lemmas 5 and 6 of \cite{ndrr}
and are given in the Appendix for completeness.

\section{Spectral divergence rates}
\label{spectral}
In the quantum information spectrum approach one defines spectral divergence
rates, defined below, which can also be viewed as generalizations of the quantum relative entropy.
\begin{definition}
Given a sequence of states $\hrho=\{\rho_n\}_{n=1}^\infty$ and a sequence of positive operators
$\hsigma=\{\sigma_n\}_{n=1}^\infty$,
the quantum spectral sup-(inf-)divergence rates are defined in terms
of the difference operators $\Pi_n(\gamma) = \rho_n - 2^{n\gamma}\sigma_n$ as
\begin{align}
\overline{D}(\hrho \| \hsigma) &:= \inf \Big\{ \gamma : \limsup_{n\rightarrow \infty} \mathrm{Tr}\big[ \{ \Pi_n(\gamma) \geq 0 \} \Pi_n(\gamma) \big] = 0 \Big\} \label{odiv} \\
\underline{D}(\hrho \| \hsigma) &:= \sup \Big\{ \gamma : \liminf_{n\rightarrow \infty} \mathrm{Tr}\big[ \{ \Pi_n(\gamma) \geq 0 \}
\Pi_n(\gamma) \big] = 1 \Big\} \label{udiv}
\end{align}
respectively.
\end{definition}
Although the use of sequences of states
allows for immense freedom in choosing them,
there remain a number of basic properties of the quantum spectral divergence
rates that hold for all sequences. These are stated and proved in
\cite{bd1}.  In the {i.i.d.} case the sequence is
generated from product states $\rho = \{ \varrho^{\otimes n}
\}_{n=1}^{\infty}$, which is used to relate the spectral entropy rates for the
sequence $\rho$ to the entropy of a single state $\varrho$.

Note that the above definitions of the spectral divergence rates
differ slightly from those originally given in (38) and (39) of
\cite{hayashi03}. However, they are equivalent, as stated in the
following two propositions \footnote{Note that in \cite{bd1} and 
\cite{hayashi03}, the logarithm was taken to base $e$,
whereas here we take the logarithm to base $2$.}. For their proofs see
\cite{bd1} or \cite{ndrr}. 
\begin{proposition}
\label{equiv_sup}
The spectral sup-divergence rate $\overline{D}(\hrho\| \hsigma)$ is equal to
\begin{equation}
\overline{\mathcal{D}}(\hrho\|\hsigma) := \inf \Big\{ \alpha : \limsup_{n\rightarrow
\infty} \mathrm{Tr}\big[ \{ \rho_n \geq 2^{n\alpha}\sigma_n \} \rho_n \big] = 0
\Big\}
\label{odiv2}
\end{equation}
which is the original definition of the spectral sup-divergence rate.
Hence the two definitions are equivalent.
\end{proposition}

\begin{proposition}
\label{equiv_inf}
The spectral inf-divergence rate $\underline{D}(\hrho \| \hsigma)$ is equivalent
to
\begin{equation}
\underline{\mathcal{D}}(\hrho\| \hsigma) = \sup \Big\{ \alpha :
\liminf_{n\rightarrow \infty} \mathrm{Tr}\big[ \{ \rho_n \geq 2^{n\alpha}\sigma_n
\} \rho_n \big] = 1 \Big\}
\label{udiv2}
\end{equation}
which is the original definition of the spectral inf-divergence rate.
Hence the two definitions are equivalent.
\end{proposition}

Despite these equivalences, it is useful to use the definitions
\reff{odiv} and \reff{udiv} for the divergence rates as they
allow the application of Lemmas \ref{lem1}, \ref{lem2} and \ref{lemmaCPT} in deriving
various properties of these rates.

The spectral generalizations of the von Neumann entropy, the conditional
entropy and the mutual information can all be expressed as spectral divergence rates with appropriate
substitutions for the sequence of operators $\hsigma = \{ \sigma_n
\}_{n=1}^{\infty}$.

\subsection{Definition of spectral entropy rates}

Consider a sequence of Hilbert spaces $\{{\cal{H}}_n\}_{n=1}^\infty$, with
${\cal{H}}_n = {\cal{H}}^{\otimes n}$.
For any sequence of states $\hrho=\{\rho_n\}_{n=1}^\infty$, with $\rho_n$ being a density matrix acting in the
Hilbert space ${\cal{H}}_n$, the sup- and inf- spectral entropy rates are defined as follows:
\begin{align}
\overline{S}(\hrho) &= \inf \Big\{ \gamma : \liminf_{n\rightarrow \infty} \mathrm{Tr}\big[ \{ \rho_n \geq 2^{-n\gamma} I_n\} \rho_n \big] = 1 \Big\} \label{os} \\
\underline{S}(\hrho) &= \sup \Big\{ \gamma : \limsup_{n\rightarrow \infty} \mathrm{Tr}\big[ \{ \rho_n \geq 2^{-n\gamma} I_n\} \rho_n\big] = 0 \Big\}.
\label{us}
\end{align}
Here $I_n$ denotes the identity operator acting in ${\cal{H}}_n$.
These are obtainable from the spectral divergence rates as follows [see \cite{bd1}:
\be
\overline{S}(\hrho)= - \underline{D} (\hrho|| \widehat{I})\, \,;\,\,\underline{S}(\hrho)= - \overline{D} (\hrho|| \widehat{I}),
\label{spec}\ee
where $\widehat{I} = \{I_n\}_{n=1}^\infty$ is a sequence of identity operators.

It is known that the spectral entropy rates of $\hrho$ are related to the von Neumann entropies of the states $\rho_n$
as follows:
\be
\underline{S}(\hrho) \le \liminf_{n\rightarrow \infty} \frac{1}{n} S(\rho_n) \le \limsup_{n\rightarrow \infty} \frac{1}{n} S(\rho_n) \le \overline{S}(\hrho).
\ee

Moreover, for a sequence of product states $\hrho=\{\rho^{\otimes n}\}_{n=1}^\infty$:
\be
\underline{S}(\hrho)= \lim_{n\rightarrow \infty} \frac{1}{n} S(\rho_n)= \overline{S}(\hrho).
\ee


For sequences of bipartite states $\hrho = \{\rho_n^{AB}\}_{n=1}^\infty$,
with $\rho_n^{AB} \in {\cal{B}}\left(({\cal{H}}_A \otimes
{\cal{H}}_B)^{\otimes n}\right)$, the conditional spectral entropy rates
are defined as follows:
\bea
\overline{S}(A|B) &:=& -\underline{D}(\hrho^{AB}\| \hI^{A}\otimes \hrho^B)
\label{ocond};\\
\underline{S}(A|B) &:=& -\overline{D}(\hrho^{AB}\| \hI^{A}\otimes \hrho^B)
.
\label{ucond}
\eea
In the above,
$\hI^{A}=\{I^A_n\}_{n=1}^\infty$ and $\hrho^{A}=\{\rho^A_n\}_{n=1}^\infty$,
with $I^A_n$ being the identity operator acting in
${\cal{H}}_A^{\otimes n}$
and $\rho^A_n = \mathrm{Tr}_B \rho^{AB}_n$, the partial trace
being taken on the Hilbert space ${\cal{H}}_B^{\otimes n}$.

Similarly, the mutual information rates are given by
\bea
\overline{S}(A:B) &:=& \overline{D}(\hrho^{AB}\| \hrho^{A}\otimes \hrho^B)
\label{omutual};\\
\underline{S}(A|B) &:=& \underline{D}(\hrho^{AB}\| \hrho^{A}\otimes \hrho^B)
.
\label{umutual}
\eea
These spectral entropy rates have several interesting properties (see e.g.\cite{bd1})
and also have the operational significance of being related to the optimal
rates of protocols (see the discussion in the Introduction and the references
quoted there).

\section{Relation between spectral divergence rates and smooth min- and max- relative entropies}
\label{thms}
In this section we prove the relations between the spectral divergence rates
and the smooth relative entropies. As mentioned in the Introduction, the 
proofs are entirely self-contained, relying only on the definitions 
of the entropic quantities involved, and the lemmas stated in Section \ref{math}.
\subsection{Relation between $\overline{D}(\hrho|\hsigma)$ and the smooth max-relative entropy}
\begin{theorem}
\label{thm_divmax}
Given a sequence of bipartite states $\hrho=\{\rho_n\}_{n=1}^\infty$, and
a sequence of positive operators $\hsigma=\{\sigma_n\}_{n=1}^\infty$, where
$\rho_n, \sigma_n\in {\cal{B}}\bigl({\cal{H}}^{\otimes n} \bigr)$,
the sup-spectral divergence rate $\oD(\hrho\|\hsigma)$, defined by \reff{odiv} (or equivalently
by \reff{odiv2})
, satisfies
\be
\oD(\hrho\|\hsigma)= \lim_{\eps \rightarrow 0} \limsup_{n \rightarrow \infty}
\frac{1}{n} D^\eps_{\max}(\rho_n||\sigma_n),
\label{main4}
\ee
where $ D^\eps_{\max}(\rho_n||\sigma_n)$ is the smooth max-entropy of the state
$\rho_n$ of the sequence $\hrho$, and the operator $\sigma_n$ of the sequence 
$\hsigma$.
\end{theorem}
\begin{proof}

We first prove the bound
\be
\oD(\hrho||\hsigma) \ge \lim_{\epsilon \rightarrow 0}\limsup_{n \rightarrow \infty}
\frac{1}{n} D^\eps_{\max}(\rho_n||\sigma_n),
\label{main44}
\ee

Let $\delta > 0$ be arbitrary but fixed, and define
\begin{equation}
  \gamma := \oD(\hrho||\hsigma) +\delta.
\end{equation}
Then from Proposition \ref{equiv_sup} it follows that
\begin{equation}
  \limsup_{n \to \infty} \tr\bigl[\{\rho^{AB}_n \geq 2^{n \gamma}  \sigma_n^{AB}\} \rho_n^{AB} \bigr] = 0 \ .
\end{equation}
In particular, for any $\eps > 0$ there exists $n_0 \in \bbN$ such that for all $n \geq n_0$,
\bea
  \tr\bigl[\{\rho^{A B}_n > 2^{n \gamma}  \sigma_n^{AB}\} \rho_n^{AB} \bigr]
&\leq &
  \tr\bigl[\{\rho^{A B}_n \geq 2^{n \gamma}  \sigma_n^{AB}\} \rho_n^{AB} \bigr]\nonumber\\
&<&
  \frac{\eps^2}{8} \ .
\eea

Using Lemma~\ref{lem6n} we then infer that for all $n \geq n_0$
\begin{equation}
D_{\max}^\eps(\rho_n^{A B} ||  \sigma_n^{AB})
\leq
  n \gamma
\end{equation}
and, hence
\begin{equation}
  \limsup_{n \to \infty} \frac{1}{n} D_{\max}^\eps(\rho_n^{A B} ||  \sigma_n^{AB})
\leq
  \gamma \ .
\end{equation}
Since this holds for any $\eps > 0$, we conclude
\begin{equation}
  \lim_{\eps \to 0} \limsup_{n \to \infty} \frac{1}{n} D_{\max}^\eps(\rho_n^{A B} ||  \sigma_n^{AB})
\leq
  \gamma = \oD(\hrho||\hsigma) + \delta \ .
\end{equation}
The assertion \reff{main44} then follows because this holds for any arbitrary $\delta > 0$.


We next prove the bound
\be
\oD(\hrho||\sigma) \le \lim_{\epsilon \rightarrow 0}\limsup_{n \rightarrow \infty}
\frac{1}{n} D^\eps_{\max}(\rho_n||\sigma_n),
\label{main66}
\ee
Let $\ttrho_{n,\eps}$ be the operator for which
\be
D_{\max} (\ttrho_{n,\eps}||\sigma_n) = \inf_{\orho \in B^\eps(\rho_n)}D_{\max} (\orho||\rho_n) = D^\eps_{\max}(\rho_n||\sigma_n).
\label{orhocond}
\ee
This implies, in particular, that for any $\lambda$ for which
$\log \lambda \ge D^\eps_{\max}(\rho_n||\sigma_n)$, we have
\be \tr\bigl[ \{ \ttrho_{n,\eps} \ge \lambda \sigma_n\} (\ttrho_{n,\eps}
- \lambda \sigma_n)\bigr]=0.\label{cnd1}
\ee
For any real constant $\gamma>0$, let us define the projection operator
 \be P_n^\gamma := \{\rho_n\ge 2^{n\gamma} \sigma_n\}.
\ee
Note that
\bea
\tr\bigl[ P_n^\gamma \rho_n\bigr]
&=& \tr\bigl[ P_n^\gamma\ttrho_{n,\eps}\bigr]
+  \tr\bigl[ P_n^\gamma( \rho_n- \ttrho_{n,\eps})\bigr]
\nonumber\\
&\le&  \tr\bigl[P_n^\gamma (\ttrho_{n,\eps}
- 2^{n\alpha} \sigma_n)\bigr]+2^{n\alpha} \tr\bigl[ P_n^\gamma \sigma_n\bigr] \nonumber\\
&& \,\,+\tr\bigl[\{\rho_n\ge \ttrho_{n,\eps}\}( \rho_n- \ttrho_{n,\eps})\bigr]
\nonumber\\
&\le & \tr\bigl[ \{ \ttrho_{n,\eps} \ge 2^{n\alpha} \sigma_n\} (\ttrho_{n,\eps}
- 2^{n\alpha} \sigma_n)\bigr]\nonumber\\
&&\,\,+ 2^{n(\alpha - \gamma)} + \eps
\label{long4}
\eea
In the above we have made use of Lemma \ref{lem1}, Lemma \ref{lem2} and Corollary \ref{cor1}.

Let $\lambda := 2 ^{n\alpha}$ and choose $\log \lambda = D^\eps_{\max}(\rho_n||\sigma_n) + \delta/2$ for any arbitrary
$\delta >0$. Further let us choose $\gamma = \alpha + \delta/2$.
Then, by \reff{cnd1}, the first term on the right hand side
of \reff{long4} vanishes. Moreover, the second term also goes to zero as 
$n \tends \infty$. Therefore, for $n$ large enough and any $\delta^{'}>0$, in the limit $\eps \tends 0$,
we must have that
\be\tr(P_n^\gamma \rho_n) \le \delta^{'}.\ee
This together with Proposition \ref{equiv_sup} implies that $\gamma \ge \oD(\hrho||\hsigma)$. The required bound \reff{main66} 
then follows from the choice of the parameters $\alpha$ and $\gamma$.
\end{proof}

\subsection{Relation between $\underline{D}(\hrho|\hsigma)$ and the smooth min-relative entropy}
\begin{theorem}
\label{thm_divmin}
Given a sequence of bipartite states $\hrho=\{\rho_n\}_{n=1}^\infty$, and
a sequence of positive operators $\hsigma=\{\sigma_n\}_{n=1}^\infty$, where
$\rho_n, \sigma_n\in {\cal{B}}\bigl({\cal{H}}^{\otimes n} \bigr)$,
the inf-spectral divergence rate $\uD(\hrho\|\hsigma)$, defined by \reff{udiv} (or equivalently by
\reff{udiv2}), satisfies
\be
\uD(\hrho\|\hsigma)= \lim_{\eps \rightarrow 0} \liminf_{n \rightarrow \infty}
\frac{1}{n} D^\eps_{\min}(\rho_n||\sigma_n),
\label{main5}
\ee
where $ D^\eps_{\min}(\rho_n||\sigma_n)$ is the smooth min-relative entropy of the state
$\rho_n$ of the sequence $\hrho$ and the operator $\sigma_n$ of the sequence 
$\hsigma$.
\end{theorem}
\begin{proof}
From the definition \reff{udiv2} of $\uD(\hrho\|\hsigma)$ it follows that for any $\gamma \le \uD(\hrho\|\hsigma)$ and any $\delta >0$, for $n$ large enough
\be\tr \big[P_n^\gamma  \rho_n\big] > 1 - \delta,\label{one3} \ee
where $P_n^\gamma := \{ \rho_n \ge 2^{n\gamma} \sigma_n\}$.

 For {{any}} given $\alpha >0$, choose $\gamma= \uD(\hrho\|\hsigma) - \alpha$, and let
\be {{\ttrho}}_{n,\gamma} := P_n^\gamma  \rho_n P_n^\gamma
\label{otrho2}
\ee
Then using (\ref{one3})
and the ``Gentle measurement lemma'', Lemma \ref{gm}, we infer that, for $n$ large enough,
${{\ttrho}}_{n,\gamma} \in B^\eps(\rho_n)$ with
$\eps = 2{\sqrt{\delta}}$. Let ${\pi}_{n,\gamma} $ denote the projection
onto the support of ${{\ttrho}}_{n,\gamma} $.

We first prove bound
\be
\lim_{\eps \rightarrow 0} \liminf_{n \rightarrow \infty}
\frac{1}{n} D^\eps_{\min}(\rho_n||\sigma_n)\ge \uD(\hrho\|\hsigma).
\label{uf1}
\ee

For $n$ large enough,
\bea
D_{\min}^\eps( \rho_n||\sigma_n)
&=& \sup_{\orho_n\in B^\eps(\rho_n)}
D_{\min}(\orho_n||\sigma_n)\nonumber\\
&\ge & D_{\min}(\ttrho_{n,\gamma}||\sigma_n)\nonumber\\
&=& - \log \tr\bigl( \pi_{n,\gamma}\sigma_n)\nonumber\\
&\ge & - \log \tr\bigl( P_{n,\gamma}\sigma_n) \ge n\gamma.
\label{new1}
\eea
The last inequality in \reff{new1} follows from Lemma \ref{lem2}. Hence, for $n$ large
enough,
\be
\frac{1}{n} D_{\min}^\eps( \rho_n||\sigma_n)
\ge \gamma = \uD(\hrho||\hsigma) - \alpha,
\ee
and since $\alpha$ is arbitrary, we obtain the desired bound \reff{uf1}.

To complete the proof of Theorem \ref{thm_divmin}, we assume that
\be
\lim_{\eps \rightarrow 0} \liminf_{n \rightarrow \infty}
\frac{1}{n} D_{\min}^\eps( \rho_n||\sigma_n)>\uD(\hrho||\hsigma),
\label{ass2}
\ee
and prove that this leads to a contradiction.

Let $\ttrho_{n,\eps}$ be the operator for which
\be
D_{\min}^\eps( \rho_n||\sigma_n)= D_{\min}( \ttrho_{n,\eps}||\sigma_n)
= - \log \tr\bigl(\tpi_{n,\epsilon}\sigma_n \bigr),
\ee
where $\tpi_{n,\eps}$ is the projection onto the support of
$\ttrho_{n,\eps}$.

Note that
\bea
\tr(\tpi_{n,\eps}  \rho_n) &=&
\tr \bigl[ \tpi_{n,\eps}\bigl((\rho_n - \ttrho_{n,\eps})
+ \ttrho_{n,\eps}\bigr)\bigr]\nonumber\\
&=& \tr \bigl[\tpi_{n,\eps}(\rho_n - \ttrho_{n,\eps})\bigr] + \tr \ttrho_{n,\eps}\nonumber\\
&\ge& \tr \bigl[ \{\rho_n \le \ttrho_{n,\eps}\} (\rho_n - \ttrho_{n,\eps})\bigr] + \tr \bigl[\ttrho_{n,\eps}\bigr]\nonumber\\
&\ge & - \eps + 1 - \eps = 1- 2\eps.
\label{new2}
\eea
We arrive at the second last line of \reff{new2} using Lemma \ref{lem1}. 
The last inequality is obtained by using the fact that $\ttrho_{n,\eps} \in B^\eps(\rho_n)$, and the bound
$$\tr \bigl[ \{\rho_n \le \ttrho_{n,\eps}\} (\rho_n - \ttrho_{n,\eps})\bigr] \ge - \eps,$$
which arises from the fact that $\ttrho_{n,\eps} \in B^\eps(\rho_n)$.

Define,
\bea
  \beta_\eps &:=&  \liminf_{n \to \infty}\Bigl[ \frac{1}{n} 
D_{\min}( \ttrho_{n,\eps}||\sigma_n)\Bigr]\nonumber\\
&:=&  \liminf_{n \to \infty}\Bigl[- \frac{1}{n} \log \tr\bigl(\tpi_{n,\epsilon}\sigma_n \bigr)\Bigr]\nonumber\\ 
\eea
and $$\gamma := \lim_{\eps \rightarrow 0} \beta_\eps.$$
Obviously, $\beta_\eps \ge \gamma$.
Note that the assumption \reff{ass2} is equivalent to the 
assumption $\gamma >\uD(\hrho||\hsigma). $
Let $\gamma_0$ be such that 
\be
\beta_\eps > \gamma_0 > \uD(\hrho||\hsigma).
\ee 
Then, for any fixed
$\eps > 0$ there exists $n_0 \in \mathbb{N}$ such that for all $n \geq
n_0$
\be
\frac{1}{n} D_{\min}( \ttrho_{n,\eps}||\sigma_n) \ge \beta_\eps.
\ee
The above inequality can be rewritten as
\be
\tr\bigl(\tpi_{n,\epsilon}\sigma_n \bigr) \leq 2^{-n\beta_\eps} 
\label{use}
\ee

Using \reff{use} we obtain the following:
\bea
\tr(\tpi_{n,\eps} \rho_n)  &=& \tr\bigl[\tpi_{n,\eps} (\rho_n - 2^{n\gamma_0}\sigma_n)\bigr] 
  + 2^{n\gamma_0} \tr \bigl[\tpi_{n,\eps}\sigma_n\bigr]\nonumber\\
&\le & \tr\bigl[\{\rho_n \ge 2^{n\gamma_0} \sigma_n\}(\rho_n - 2^{n\gamma_0} \sigma_n) \bigr]
 +2^{n\gamma_0} 2^{-n \beta_\eps}\nonumber\\
&\le &  \tr\bigl[\{\rho_n \ge 2^{n\gamma_0} \sigma_n\}(\rho_n - 2^{n\gamma_0} \sigma_n) \bigr]
 + 2^{-n(\beta_\eps - \gamma_0)}\nonumber\\
\label{upb7}
\eea
The second term on the right hand side of
\reff{upb7} tends to zero asymptotically in $n$, since $\delta_\eps>0$. However, the first term does not tend to $1$,
since $\gamma_0 >\uD(\hrho||\hsigma)$ by assumption. Hence we obtain the bound
\be
\tr(\tpi_{n,\eps} \rho_n) < 1 - c_0,
\label{compress2}
\ee
for some constant $c_0>0$. This contradicts \reff{new2} in the limit $\eps \rightarrow 0$.
\end{proof}

\section{Appendix}
\label{lempfs} 
\noindent
{\bf{Proof of Lemma \ref{lem5n}}}

\begin{proof} Define
  \begin{align*}
  \alpha_{A B} & := 2^{\lambda} \cdot \sigma_{AB} \\
  \beta_{A B} & := 2^{\lambda} \cdot \sigma_{AB} + \Delta_{A B} \ .
  \end{align*}
 and
  \[
    T_{A B} := \alpha_{A B}^\half \beta_{A B}^{-\half} \ .
  \]
  Let $\ket{\Psi} = \ket{\Psi}_{A B R}$ be a purification of $\rho_{A B}$ and let $\ket{\Psi'} := T_{A B} \otimes \id_R \ket{\Psi}$ and $\rho'_{A B} := \tr_R(\proj{\Psi'})$.

Note that
\begin{align*}
  \rho'_{A B}
& =
  T_{A B} \rho_{A B} T_{A B}^{\dagger} \\
& \leq
  T_{A B} \beta_{A B} T_{A B}^{\dagger} \\
& =
  \alpha_{A B}
=
  2^{\lambda} \cdot \sigma_{AB} \ ,
\end{align*}
which implies $D_{\max}(\rho'_{A B}\|\sigma_{AB}) \le \lambda$. It thus remains to be shown that
\begin{equation} \label{eq:distbound}
  \| \rho_{A B} - \rho'_{A B} \|_1
\leq
  \sqrt{8 \tr(\Delta_{A B})}
 \ .
\end{equation}

  We first show that the Hermitian operator
  \[
    \bar{T}_{A B} := \frac{1}{2} (T_{A B} + T_{A B}^\dagger) \ .
  \]
  satisfies
  \begin{equation} \label{eq:Tleqid}
    \bar{T}_{A B} \leq \id_{A B} \ .
 \end{equation}
 For any vector $\ket{\phi} = \ket{\phi}_{A B}$,
 \begin{align*}
   \| T_{A B} \ket{\phi} \|^2
 & =
   \bra{\phi} T_{A B}^\dagger T_{A B} \ket{\phi}
=
   \bra{\phi} \beta_{A B}^{-\half} \alpha_{A B} \beta_{A B}^{-\half} \ket{\phi} \\
 & \leq
   \bra{\phi} \beta_{A B}^{-\half} \beta_{A B} \beta_{A B}^{-\half} \ket{\phi}
 =
   \| \ket{\phi} \|^2
  \end{align*}
 where the inequality follows from $\alpha_{A B} \leq \beta_{A B}$.
Hence, for any vector $\ket{\phi}$,
 \begin{align*}
  \| \bar{T}_{A B} \ket{\phi} \|
& \leq
   \frac{1}{2} \| T_{A B} \ket{\phi} + T_{A B}^{\dagger} \ket{\phi} \| \\
 & \leq
   \frac{1}{2} \| T_{A B} \ket{\phi} \| + \frac{1}{2} \| T_{A B}^{\dagger} \ket{\phi} \|
 \leq
   \| \ket{\phi} \| \ ,
 \end{align*}
 which implies~\eqref{eq:Tleqid}.

  We now determine the overlap between $\ket{\Psi}$ and $\ket{\Psi'}$,  \begin{align*}
    \spr{\Psi}{\Psi'}
   & =
    \bra{\Psi}  T_{A B} \otimes \id_R \ket{\Psi} \\
   & =
    \tr(\proj{\Psi} T_{A B} \otimes \id_R)
   =
    \tr(\rho_{A B} T_{A B}) \ .
 \end{align*}
 Because $\rho_{A B}$ has trace one, we have
\begin{align*}
    1 - |\spr{\Psi}{\Psi'}|
  & \leq
    1- \Re \spr{\Psi}{\Psi'}
  =
    \tr\bigl(\rho_{A B} (\id_{A B} - \bar{T}_{A B}) \bigr) \\
   & \leq
     \tr\bigl(\beta_{A B}  (\id_{A B} - \bar{T}_{A B})\bigr) \\
   & =
     \tr(\beta_{A B}) - \tr(\alpha_{A B}^{\half} \beta_{A B}^{\half}) \\
   & \leq
     \tr(\beta_{A B}) - \tr(\alpha_{A B})
   =
     \tr(\Delta_{A B}) \ .
  \end{align*}
  Here, the second inequality follows from the fact that, because of~\eqref{eq:Tleqid}, the operator $\id_{AB} - \bar{T}_{A B}$ is positive, and $\rho_{A B} \leq \beta_{A B}$. The last inequality holds because $\alpha_{A B}^{\half} \leq \beta_{A B}^{\half}$, which is a consequence of the operator monotonicity of the square root (Proposition V.1.8 of \cite{bhatia}).

Using \reff{fidelity} and the fact that the fidelity between two pure states is given by their overlap, we find
\begin{align*}
  \| \proj{\Psi} - \proj{\Psi'} \|_1
& \leq
  2 \sqrt{2(1-| \spr{\Psi}{\Psi'} |)} \\
& \leq
  2 \sqrt{2 \tr(\Delta_{A B})}
\leq
  \eps \ .
\end{align*}
Inequality~\eqref{eq:distbound} then follows because the trace distance can only decrease when taking the partial trace.
\end{proof}

\noindent
{\bf{Proof of Lemma \ref{lem6n}}}

\begin{proof}
  Let $\Delta^+_{A B}$ and $\Delta^-_{A B}$ be mutually orthogonal positive operators such that
  \[
    \Delta^+_{A B} - \Delta^-_{A B} = \rho_{A B} - 2^{\lambda} \sigma_{AB} \ .
  \]
  Furthermore, let $P_{A B}$ be the projector onto the support of $\Delta^+_{A B}$, i.e.,
  \[
    P_{A B} = \{\rho_{A B} > 2^{\lambda} \sigma_{AB} \} \ .
  \]
  We then have
  \begin{align*}
    P_{A B} \rho_{A B}  P_{A B}
  & =
    P_{A B} (\Delta^+_{A B} + 2^{\lambda} \sigma_{AB} - \Delta^-_{A B}) P_{A B} \\
  & \ge 
    \Delta^{+}_{A B}
  \end{align*}
  and, hence,
  \[
    \sqrt{8 \tr(\Delta^{+}_{A B})}
  \le
    \sqrt{8 \tr(P_{A B} \rho_{AB})\bigr)} = \eps \ .
  \]
  The assertion now follows from Lemma~\ref{lem5n} because
  \[
    \rho_{A B} \leq 2^{\lambda} \sigma_{AB} + \Delta^+_{A B} \ .
  \]
\end{proof}

\section*{Acknowledgements} The author is very grateful to Milan Mosonyi for
carefully reading the paper and pointing out an error in the first version. She
would also like to thank Reinhard Werner and Tony Dorlas for helpful suggestions.  

\section*{Biogragraphy}
\medskip

\noindent
Nilanjana Datta received a Ph.D. degree from ETH Zurich,
Switzerland, in 1996. From 1997 to 2000, she was a 
postdoctoral researcher at the Dublin Institute of Advanced 
Studies, C.N.R.S. Marseille, and EPFL in Lausanne. In 2001
she joined the University of Cambridge, as a Lecturer in
Mathematics of Pembroke College, and a member of the Statistical
Laboratory, in the Centre for Mathematical Sciences. She is 
currently an Affiliated Lecturer of the Faculty of Mathematics,
University of Cambridge, and a Fellow of Pembroke College.

\end{document}